\documentclass[12pt]{article}
\usepackage{amsthm}
\usepackage{amsfonts,amsmath,amssymb,mathrsfs}
\usepackage[colorlinks,linkcolor=blue,anchorcolor=blue,citecolor=blue]{hyperref}
\usepackage{flafter,graphicx,xcolor}
\usepackage{multicol}
\usepackage{pythonhighlight}
\usepackage{blindtext}
\usepackage{longtable}
\usepackage{makecell}
\usepackage[top=2cm, bottom=2cm, left=2cm, right=2cm]{geometry}

\newtheorem{theorem}{Theorem}[section]
\newtheorem{lemma}[theorem]{Lemma}

\newtheorem{corollary}[theorem]{Corollary}

\theoremstyle{definition}

\newtheorem{example}[theorem]{Example}

\newtheorem{remark}[theorem]{Remark}
 
\newcommand\F{\mathbb{F}} 
\newcommand\Z{\mathbb{Z}}
\newcommand\C{\mathcal{C}}
\newcommand\lcm{\mathrm{lcm}} 
\newcommand\ord{\mathrm{ord}}

\numberwithin{equation}{section}
\numberwithin{table}{section}

\begin{document}
  
\title{Constacyclic codes with best-known parameters} 

\author{Zekai Chen and Min Sha}

\date{}

\maketitle

\abstract{
In this paper, we construct several infinite families of $q$-ary
constacyclic codes over a finite field $\F_q$ with length $n$, dimension around $n/2$, and minimum distance at least $cn/\log_q n$ for some positive constant $c$. They contain many constacyclic codes with optimal, or almost-optimal, or best-known parameters. 
We also consider constacyclic codes of various lengths. 
}

\let\thefootnote\relax\footnote{
Z. Chen's research was supported by the Scientific Research Innovation Project of Graduate School of South China Normal University.
M. Sha's research was supported by the Guangdong Basic and Applied Basic Research Foundation (No. 2025A1515010635). 

Zekai Chen and Min Sha are with School of Mathematical Sciences, South China Normal University, Guangzhou, Guangdong, 510631, China. (email: chenzk@m.scnu.edu.cn, min.sha@m.scnu.edu.cn)
}

\section{Introduction}

\subsection{Background} 
Let $\F_q$ be the finite field with $q$ elements and $\F_q^* = \F_q \setminus \{0\}$. 
A $q$-ary $[n,k,d]$ linear code $\C$ is a $k$-dimensional linear subspace of $\F_q^n$ with minimum distance $d$. 
Here, $n$ is called the length of $\C$. 
The code $\C$ is said to be distance-optimal (or, optimal, for short)
if there is no $q$-ary $[n, k,\ge d+1]$ linear code. 
In addition, $\C$ is said to be distance almost-optimal (or, almost-optimal, for short)
if there is a  distance-optimal $q$-ary $[n, k,d+1]$ linear code. 
Moreover,  $\C$ is said to be best-known
if $d$ is the maximal minimum distance of all known $q$-ary linear codes with length $n$ and dimension $k$. 

Let $\lambda \in \F_q^*$. 
The code $\C$ is called a $\lambda$-constacyclic code if $(\lambda c_{n-1}, c_0, \ldots, c_{n-2})\in \C$ 
for any $(c_0,c_1, \ldots, c_{n-1}) \in \C$. Define the map 
\begin{align*}
& \phi: \, \F_q^n \to \F_q[x]/\langle x^n - \lambda \rangle \\
& (c_0,c_1, \ldots, c_{n-1}) \mapsto c_0 + c_1 x +  \cdots + c_{n-1} x^{n-1}.
\end{align*}
It is well-known that every ideal of the residue class ring $\F_q[x]/\langle x^n - \lambda \rangle$ is principal. 
Then, the code $\C$ is called a $\lambda$-constacyclic code if and only if $\phi(\C)$ is an ideal of the ring $\F_q[x]/\langle x^n - \lambda \rangle$. 
Moreover, we identify $\C$ with $\phi(\C)$, and then we can write 
 $\C = \langle g(x) \rangle$, where $g(x)$ is monic and has the smallest degree among all the non-zero codewords of $\C$. 
Here, $g(x)$ is in fact a factor of $x^n - \lambda$ and is called the generator polynomial of $\C$. 

In particular, when $\lambda =1$, $\lambda$-constacyclic codes are classical cyclic codes;
 and when  $\lambda =-1$ and $q$ is odd, $\lambda$-constacyclic codes are negacyclic codes. 
We refer to \cite{HP, MS, PW} for more details about linear codes and cyclic codes. 

Constacyclic codes have attracted extensive attentions due to their performance and applications. 
For example,  several new constacyclic codes that improve the minimum distance
of currently best known linear codes were found in \cite{AAHP, AM}, 
 and several infinite classes of MDS constacyclic codes were constructed in \cite{DP, KS, WDLZ}. 
In addition, constacyclic codes have important applications 
in constructing other kinds of codes, such as locally
repairable codes \cite{CFXF}, symbol-pair codes \cite{CLL},  and quantum codes \cite{CLZ, KZL}.
We refer to \cite{SDW} for more details and references about constacyclic codes.

\subsection{Related works}

Constructing cyclic codes or more general constacyclic codes with optimal or best-known parameters 
is interesting and important. 
Many optimal or best-known  codes over small fields with short lengths are cyclic codes; see \cite{Grassl}. 
Some infinite families of optimal ternary cyclic codes with length $3^m -1$
and minimum distance four were constructed in \cite{DH}. 
Several other infinite families of optimal ternary cyclic codes with length $3^m -1$
and minimum distance four or five were constructed in \cite{LLHDT}. 
Some infinite families of optimal ternary negacyclic codes with minimum distance four or five were constructed in \cite{SD}. 
Infinitely many families of optimal binary cyclic codes with
minimum distance six and an infinite family of quaternary cyclic codes with minimum
distance four were constructed in \cite{CW}. 
The first infinite family of binary optimal BCH codes
with minimum distance 8 was constructed in \cite{XCY}. 
 Three infinite families of optimal repeated-root cyclic codes with minimum distance three or four were constructed 
in \cite{CXD}. 
We refer to \cite{CW, CXD, DDR, LC, SD, WYZZ, XCY, ZHLC, ZKZL} and the references therein for more optimal or best-known codes constructed as cyclic, negacyclic, or quasi-cyclic codes. 

On the other hand, it is a long-standing open problem in coding theory that whether 
there exists an asymptotically good sequence of cyclic codes (see \cite{MW}).
So, it is interesting to construct infinite families of cyclic, negacyclic
or constacyclic codes whose dimensions and minimum distances are both large. 

In recent years, researchers have paid much attention in constructing infinite families of such codes with parameters $[n, k, d]$ 
such that $k$ is around $n/2$ and $d$ is at least $c\sqrt{n}$ for some positive constant $c$ depending only on $q$  
(this lower bound is called \textit{the square-root-like lower bound}). 

Some binary quadratic residue codes with  parameters $[n,(n+1)/2, d]$, where $d \ge \sqrt{n}$ and $n$ is a prime with $n \equiv \pm 1 \pmod{8}$, were constructed in   \cite[Section 6.6]{HP}. 
Two infinite families of binary cyclic codes with length $2^{m}-1$ and dimension near $2^{m-1}$ and 
with a square-root-like lower bound on the minimum distance were constructed in \cite{TD}. 
Recently, the work in \cite{TD} was extended to the 4-ary case in \cite{STKS} and the $2^s$-ary case in \cite{LSTS}. 
We refer to \cite{LGLS, LLD, LQS, Sun, SL} for more infinite families of binary cyclic codes with a square-root-like lower bound, 
\cite{CDLS, LQS} for infinite families of ternary cyclic codes with a square-root-like lower bound, 
and  \cite{CD, SWD, XL, ZKL} for  infinite families of  $q$-ary cyclic codes with a square-root-like lower bound. 
In addition, several infinite families of  $q$-ary negacyclic codes and constacyclic codes were 
constructed in \cite{CSXCD, SDW, XCDS} with a square-root-like lower bound.

In particular, in \cite{SLD} the first infinite family of binary $[n, (n+1)/2, d]$ cyclic codes  was constructed  with minimum distance $d$ provably much better than the square-root-like bound, where $n=2^p -1$, $p$ is an odd prime and $d \ge cn/\log_2 n$ for some positive constant $c$. 
Roughly speaking, the defining set of these codes is the union of the first half of cyclotomic cosets with size $p$. 

Motivated by \cite{SLD}, several works have been done recently for constructing more $q$-ary cyclic or negacyclic codes 
with a lower bound like $cn/\log_q n$ on their minimum distances.  
We summarize them as follows: 
\begin{itemize}
\item  In \cite{CW}, one infinite family of $q$-ary $[n, (n+1)/2,  \ge cn/\log_q n]$ negacyclic  codes with $q \equiv 3 \pmod{4}$ and $n=\frac{q^p-1}{q-1}$ was constructed, and one infinite family of $q$-ary negacyclic $[n, n/2,  \ge cn/\log_q n]$ codes with $q \equiv 1 \pmod{4}$ and $n=\frac{q^p-1}{2}$ was also constructed. 

\item In \cite{CXD},  two infinite families of repeated-root binary cyclic codes with parameters
 $[2n,k,\ge (n-1)/\log_2 n]$, where $n = 2^p-1$ and $k \ge n$, were constructed.

\item In \cite{LGLS}, two infinite families of binary $[n, (n+1)/2, \ge cn/\log_2 n]$ cyclic codes 
with $n=2^{p^2}-1$ or $2^{p_1p_2}-1$ were constructed, 
where $p_1$ and $p_2$ are two distinct odd primes. 

\item In \cite{LWS}, an improvement of the lower bound on the minimum distance in \cite{SLD} was given,  
and an infinite family of ternary $[n, (n+2)/2,  \ge cn/\log_3 n]$ cyclic codes 
was constructed, where $n=3^p-1$. 

\item  In \cite{SL}, an infinite family of binary $[2^m-1,k, \ge \lfloor 2^{m-1}/m \rfloor+2]$ cyclic codes 
was constructed for any integer $m \ge 3$, where $ k \ge 2^{m-1}-1$.  
\end{itemize}

Most recently, an infinite family of binary BCH codes was constructed in \cite{XCY} 
with length $n=(2^{p_1}-1) \cdots (2^{p_s}-1)$, 
dimension at least $(n+1)/2$, and minimum distance at least $\lceil (n-1)/(p_1 \cdots p_s) \rceil$, 
where $p_1, \ldots, p_s$ are distinct primes. 
We also refer to \cite{NLJM} for BCH codes with minimum distances
propotional to code lengths and refer to \cite{DL} for a recent survey on BCH codes.

\subsection{Our contributions}

In this paper, motivated by \cite{SLD} and generalizing the constructions in \cite{CW, CXD, LGLS, LWS, SL, SLD}, we propose the following general construction for $q$-ary constacyclic codes: 
 first we classify the cyclotomic cosets according to their sizes 
(that is, two cosets are in the same class if and only if they have the same size), and then 
we take the defining set as the union of the first half of cosets in each class 
(we sort the cosets in ascending order by their coset leaders and use the ceiling function or the floor function to achieve ``the first half"). 

In order to obtain good lower bounds for minimum distance, in this paper we focus 
on the case when there are exactly two kinds of cyclotomic cosets.  
Even in this case, the codes we construct are not always BCH constacyclic codes; see Example~\ref{ex:not-BCH}. 

Let $N$ be the number of cyclotomic cosets with maximum size. 
In the above construction, one may guess that the first $\lfloor N/2 \rfloor$ consecutive elements
are in the defining set. 
However, this does not always hold; see Example~\ref{ex:deltaN1Nm}.

The main contributions of this paper are as follows: 
\begin{itemize}
\item 
We establish the lower bound $\lfloor \frac{qN}{2(q-1)} \rfloor$ on minimum distances  
for $q$-ary constacyclic codes constructed by using the ceiling function, and the same lower bound in many situations for those codes constructed by using the floor function (see Theorem~\ref{thm:NlNm}). 
This improves the lower bounds on minimum distances in \cite[Theorems 10, 15 and 16]{CW} for some infinite families of $q$-ary negacyclic codes mentioned above. 
Moreover, for some special cases, we obtain better lower bounds on minimum distances (see Theorem~\ref{thm:NlNm2}).

\item 
We consider $q$-ary constacyclic codes of various lengths $n$: $\frac{q^p -1}{rs}$, or $\Phi_{p^{b}}(q)$, 
 or $\Phi_{p_1 p_2}(q)$, where  $p$ is a prime,  $r$ is the multiplicative order of $\lambda$, $s$ and $b$ are two positive integers, $p_1$ and $p_2$ are two distinct primes, 
and $\Phi_m(x)$ is the $m$-th cyclotomic polynomial for any integer $m \ge 1$. 
We obtain many constacyclic codes with optimal, or almost-optimal, or best-known parameters (according to \cite{Grassl}). 
We also present several codes which are better than BCH constacyclic codes.
\end{itemize}

The rest of this paper is organized as follows. In Section~\ref{sec:pre}, we present some preliminary results. 
In Section~\ref{sec:cons}, we present some general constructions for infinite families of constacyclic codes. 
In Section~\ref{sec:qp}, we construct two infinite families of constacyclic codes with length $n=\frac{q^p-1}{rs}$ for any positive integers $r, s$ such that $r \mid q-1$ and $s \mid \frac{q^p -1}{r}$. 
In Section~\ref{sec:pbq}, we construct several infinite families of constacyclic codes with length $n=\Phi_{p^b}(q)$.  
In Section~\ref{sec:p1p2}, we construct two infinite families of constacyclic codes with length $n=\Phi_{p_1p_2}(q)$.

\section{Preliminaries}   \label{sec:pre}

\subsection{Cyclotomic cosets}
Let $q$ be a power of prime,
$n$ be a positive integer with $\gcd(n,q)=1$, and 
$r \mid q-1$. 
Let 
$$
\Z_{nr}=\{1,2, \ldots, nr\}
$$
be the ring of integers modulo $nr$. 
Here we emphasize that the first element in $\Z_{nr}$ is 1 but not 0, 
and we need this to unify the statements about cyclic codes and constacyclic codes. 

For any $i \in \Z_{nr}$,
the $q$-cyclotomic coset of $i$
modulo $nr$ is defined by 
$$
C_i^{(q,nr)}=\{ i q^j \mod nr : \, 0 \leq j \leq l_i-1\} \subseteq \Z_{nr},
$$
where $l_i$ is the least positive integer such that $i q^{l_i} \equiv i \pmod{nr}$, that is, 
$$
l_i = |C_i^{(q,nr)}|.  
$$ 
For simplicity, in the sequel we denote $C_i = C_i^{(q,nr)}$. 

The smallest integer in $C_i$ is called the coset leader of $C_i$. Let
$\Gamma_{(q,n,r)}$ be the set of all the coset leaders, 
$$
\Gamma_{(q,n,r)}^{(1)} = \{i \in \Gamma_{(q,n,r)}: \, i \equiv 1 \pmod{r}\}, 
$$
and 
$$
Z_{n,r} = \{1+ir: \, 0 \le i \le n-1\} \subseteq \Z_{nr}. 
$$
Clearly, $Z_{n,1}= \Z_n$ if $r=1$ (this corresponds to cyclic codes), and 
$$
Z_{n,r} = \bigcup_{i \in \Gamma_{(q,n,r)}^{(1)}} C_i. 
$$
We also remark that the case $r=2$  corresponds to negacyclic codes.

Let $\beta$ be a primitive $nr$-th roots of unity
and $\lambda=\beta^{n}$. 
Then, $\lambda$ is an element in $\F_q$ with multiplicative order $r$. 
For any $i \in \Z_{nr}$,
let $M_{\beta^i}(x)$ denote the minimal polynomial of $\beta^i$ over $\F_q$,
then 
$$
M_{\beta^i}(x)=\prod_{j \in C_i} (x-\beta^j) \in \F_q[x].
$$
Besides,
$$
x^n-\lambda = x^n - \beta^n = \prod_{j \in Z_{n,r}} (x-\beta^j) = \prod_{i \in \Gamma_{(q,n,r)}^{(1)}} M_{\beta^i} (x) .
$$
Then, for a $\lambda$-constacyclic code $\C =  \langle g(x) \rangle$ with $g(x) \mid x^n - \lambda$,  
the set 
$$
T(\C) = \{i\in Z_{n,r}: \, g(\beta^i)=0 \} 
$$
is called the defining set of $\C$ with respect to $\beta$. 
In addition, the Bose distance of $\C$ is defined to be one plus the largest number of consecutive elements from $Z_{n,r}$ and  contained in $T(\C)$ (see \cite{DDZ} for some work about the Bose distance of BCH codes).

For any positive integer $m$ with $\gcd(q, m)=1$, let $\ord_m (q)$ be the the multiplicative
order of $q$ modulo $m$. 
The result in the following lemma is somehow well-known (for instance, see \cite[page 3]{WCM}).

\begin{lemma}\label{lem:liord}
  For any $i \in Z_{n, r}$,
  we have $l_i \mid \ord_{nr}(q)$. 
In particular, if we write $nr=(q^m -1)/s$ for some positive integers $m, s$, then we have $l_i \mid m$ for each $i \in Z_{n, r}$. 
\end{lemma}

For any positive integer $l$, let $N_l^{(q,n,r)}$ be the number of $q$-cyclotomic cosets modulo $nr$ 
whose coset leaders are in $\Gamma_{(q,n,r)}^{(1)}$ and sizes are equal to $l$, that is, 
$$
N_l^{(q,n,r)} = \Big| \big\{C_i: \, |C_i|=l, i \in Z_{n,r} \big\} \Big|.
$$
By Lemma~\ref{lem:liord}, if $l \nmid \ord_{nr}(q)$, then $N_l^{(q,n,r)} = 0$.

In the sequel, for simplicity and without confusion we denote 
\begin{align*}
N_l = N_l^{(q,n,r)}, \quad l \ge 1.
\end{align*}

We recall two lemmas in \cite[Lemmas 4 and 5]{WCM},
which are useful for determining $N_l$.

\begin{lemma}[\cite{WCM}] \label{lem:limidliff2} 
For any $i \in Z_{n,r}$ and any positive integer $l$, we have 
 $$l_i \mid l \iff \frac{nr}{\gcd(nr,q^l-1)} \, \Big| \, i .$$
\end{lemma}

\begin{lemma}[\cite{WCM}] \label{lem:Nl2consta}
  For each $l \mid \ord_{nr}(q)$,
  we have 
  $$
  N_l = \Big(\sum_{j \mid l} \mu(l/j) \sigma(j) \gcd(n,(q^j-1)/r) \Big)/l,
  $$
  where $\mu$ is the M\"{o}bius function and
  $$
    \sigma(j)=\begin{cases}
    0,& \gcd(r,\frac{nr}{\gcd(nr,q^j-1)})>1,\\
    1,& \gcd(r, \frac{nr}{\gcd(nr,q^j-1)})=1.
  \end{cases}
  $$
\end{lemma}

From Lemma~\ref{lem:Nl2consta}  and  using the M\"{o}bius inversion formula, 
we directly get the following formula about 
the number of elements $i \in Z_{n,r}$ with $l_i$ dividing a given integer. 

\begin{lemma}\label{lem:sumlNl}
  For any $l \mid \ord_{nr}(q)$,
  we have
  $$
  \sum_{j \mid l} j N_j =\begin{cases}
    0,& \gcd(r,\frac{nr}{\gcd(nr,q^l-1)})>1,\\
    \gcd(n,(q^l-1)/r), &\gcd(r,\frac{nr}{\gcd(nr,q^l-1)})=1.
  \end{cases}
  $$
\end{lemma}

\subsection{Coset leaders in $Z_{n,r}$}   \label{sec:CL}

For any positive integer $l$, if there exist $q$-cyclotomic cosets with $l$ elements (that is, $N_l^{(q,n,r)} > 0$)  in $Z_{n,r}$, then 
we list their coset leaders in ascending order,
and we denote the $i$-th one by $\delta_i^{(q,n,r,l)}$ for $1 \le i \le N_l^{(q,n,r)}$.

In the sequel, for simplicity and without confusion we denote 
\begin{align*}
& N_l = N_l^{(q,n,r)}, \quad l \ge 1, \\
&  \delta_i^{(l)} = \delta_i^{(q,n,r,l)}, \quad 1\le i \le N_l. 
\end{align*}

We want to get some estimates on the coset leaders $\delta_i^{(l)}$ for $1\le i \le N_l$. 
For this, we need a preparation. 

\begin{lemma}  \label{lem:ar}
For any two positive integers $a$ and $i$, the $i$-th element  not disivible by $q$ in the set 
$\{a, a+r, a+2r, \ldots \}$ $($listed in ascending order$)$ is at least 
$$
a + \left(\lceil qi/(q-1) \rceil - 2 \right) r.
$$
\end{lemma}

\begin{proof}
Denote the set $S_{a,r}=\{a, a+r, a+2r, \ldots \}$, and we list its elements in ascending order. 
Since $\gcd(q,r)=1$ (due to $r \mid q-1$),
  we have that every $q$ consecutive elements in $S_{a,r}$ form a complete residue system modulo $q$.

Assume that the $i$-th element  not disivible by $q$ in the set 
$S_{a,r}$ is $a+ (j-1)r$ for some integer $j$. 
Then,  we have either $i = j - \lfloor j/q \rfloor - 1$, or $i = j - \lfloor j/q \rfloor$. 
If $i = j - \lfloor j/q \rfloor - 1$, then we have 
$
i \le j -  (j/q - (q-1)/q) -1 ,
$
 which gives 
$
j \ge \frac{qi}{q-1} + \frac{1}{q-1}. 
$
If $i = j - \lfloor j/q \rfloor$, then we have 
$
i \le j -  (j/q - (q-1)/q),
$
 which implies
$
j \ge \frac{qi}{q-1} - 1. 
$
Hence, noticing that $j$ is an integer, we always have 
$$
j \ge \left\lceil \frac{qi}{q-1} - 1 \right\rceil  = \left\lceil \frac{qi}{q-1}\right\rceil - 1. 
$$
This yields the desired result.
\end{proof}

Now, we are ready to give a lower bound for $\delta_i^{(l)}$ with $1\le i \le N_l$. 

\begin{lemma}  \label{lem:delta1}
For any positive integer $l$, if $N_l > 0$, then for any integer $i$ with $1\le i \le N_l$, we have 
\begin{align*}
  \delta_i^{(l)} & \ge  \frac{nr(1 + (\lceil qi/(q-1) \rceil -2) r)}{\gcd(nr,q^l-1)} \\ 
& \geq  \frac{nr^2(i-1)}{\gcd(nr,q^l-1)}+\frac{nr}{\gcd(nr,q^l-1)}.
\end{align*}
\end{lemma}

\begin{proof}
First,by Lemma~\ref{lem:limidliff2} we have 
$
\frac{nr}{\gcd(nr,q^l-1)} \, \Big| \, \delta_i^{(l)}. 
$
For simplicity, we denote $t = nr/\gcd(nr,q^l-1)$, then $\delta_i^{(l)}$ is a  multiple of $t$. 
Since $r \mid q-1$, we have $t \mid n$.

Besides, since $N_l > 0$, by Lemma~\ref{lem:sumlNl} we have  $\gcd(r,t)=1$.
So, every $t$ consecutive elements in $Z_{n,r}$ form a complete residue system modulo $t$. 
If $at$ is the smallest multiple of $t$ contained in $Z_{n,r}$ for some positive integer $a \le r$, 
then the set of all the multiples of $t$ in $Z_{n,r}$ is (noticing $\gcd(r,t)=1$): 
$$
S(t) = \{at, at+rt, at+2rt, \ldots, at+jrt\}= \{at, (a+r)t, (a+2r)t, \ldots, (a+jr)t\}, 
$$
where $j$ is the largest integer such that $at+jrt \le 1+(n-1)r$. 
Hence, for our purpose it suffices to estimate the $i$-th element in $S(t)$ (listed in ascending order) which is not disivible by $q$. 

Since $\gcd(n,q)=1$ and $t \mid n$, we have $\gcd(t,q)=1$. 
So, it suffices to estimate the $i$-th element in the set $\{a, a+r, \ldots, a+jr\}$ (listed in ascending order) which is not disivible by $q$.  
Hence, combining this with Lemma~\ref{lem:ar}, we obtain 
\begin{equation*}
\begin{split}
\delta_i^{(l)}  \ge \left(a + (\lceil qi/(q-1)\rceil -2)  r \right)t   
 \ge \frac{nr(1 + (\lceil qi/(q-1)\rceil -2) r)}{\gcd(nr,q^l-1)}.
\end{split}
\end{equation*}
This completes the proof of the first lower bound in the lemma. 

The second lower bound in the lemma follows directly from the first one, because 
$$
\lceil qi/(q-1) \rceil -2 = i + \lceil i/(q-1) \rceil -2 \ge i + 1-2 =i-1. 
$$
\end{proof}

We remark that the first lower bound in Lemma~\ref{lem:delta1} is somehow optimal for the general case. 
For example, choosing $q=3, n=80, r=1$, we have $\delta_2^{(2)} = 20$, 
which coincides with the lower bounds in Lemma~\ref{lem:delta1} (with $i=2, l=2$). 

We also need to estimate $\delta_i^{(l)}$ further for some special $i$'s. 

\begin{lemma}\label{lem:deltaNm}
Write  $nr=(q^m-1)/s$ for some positive integers $m,s$. 
  Then,  if $N_m \ge 1$, we have  
  $$
\delta^{(m)}_{\lceil N_m/2 \rceil+1} \ge \delta^{(m)}_{\lfloor N_m/2 \rfloor+1} >  1 + \left(\frac{qN_m}{2(q-1)} -2\right)r. 
$$
\end{lemma}

\begin{proof}
If $N_m \ge 1$, by Lemma~\ref{lem:delta1} we get 
\begin{align*}
  \delta^{(m)}_{\lceil N_m/2 \rceil+1}  
  \ge \delta^{(m)}_{\lfloor N_m/2 \rfloor+1} 
& \ge \frac{nr(1 + (\lceil q(\lfloor N_m/2 \rfloor+1)/(q-1) \rceil -2) r)}{\gcd(nr,q^m-1)} \\ 
 & = 1 + (\lceil q(\lfloor N_m/2 \rfloor+1)/(q-1) \rceil -2) r \\
 & \ge 1 + (q(\lfloor N_m/2 \rfloor+1)/(q-1) -2) r \\
 & > 1 + \left(qN_m/(2(q-1)) -2 \right) r, 
\end{align*}
which gives the desired result.
\end{proof}

The following lemma is a generalization of \cite[Lemma 9]{SLD} 
for determining some elements which are not coset leaders. 

\begin{lemma}\label{lem:NotLeader}
 Assume $nr=(q^m-1)/s$.  
  Then, $q^{i}+qt+j$ is not a coset leader modulo $nr$ 
for any integer $i$ with $(m+1)/2 \leq i \leq \lfloor m-\log_q(s+1) \rfloor$, 
 any integer $j$ with  $1 \leq j \leq q-1$ and any integer $t$ with  
  $$0 \leq t < \frac{q^{i}-j(q^{m-i}-1)-1}{q^{m-i+1}-q}.$$
Moreover, all these $q^{i}+qt+j$ are less than $nr$ and pairwise distinct. 
\end{lemma}

\begin{proof}
 First, we have
  \begin{align*}
    (q^{i}+qt+j)q^{m-i} \equiv& (qt+j) q^{m-i}+1 \pmod{q^m-1}\\
    \equiv& (qt+j) q^{m-i}+1 \pmod{nr}.
  \end{align*}
In addition, we have 
  \begin{align*}
    (qt+j)q^{m-i}+1
    &=(qt+j)(q^{m-i}-1)+qt+j+1\\
    &<(\frac{q^{i}-j(q^{m-i}-1)-1}{q^{m-i}-1} +j)\cdot (q^{m-i}-1)+qt+j+1\\
    &=q^{i}-1+qt+j+1
    =q^{i}+qt+j,   
  \end{align*}
  and
  $$
  q^{i}+qt+j<q^i + \frac{q^{i}-j(q^{m-i}-1)-1}{q^{m-i}-1} +j=\frac{q^m-1}{q^{m-i}-1} \leq \frac{q^m-1}{s}=nr, 
  $$
where we need the condition $i \le \lfloor m - \log_q(s+1) \rfloor$. 
Hence,  $q^{i}+qt+j$ is not a coset leader modulo $nr$. 

Due to the choice of $j$, if we fix $i$, then clearly  $q^{i}+qt+j$ are pairwise distinct for various $j$ and $t$. 

Finally, notice that for any $i$ with $(m+1)/2 \leq i \leq \lfloor m - \log_q(s+1) \rfloor-1$, we have 
\begin{equation}  \label{eq:qi+1}
\begin{split}
q^i +qt + j  < q^i + q \cdot \frac{q^i - (q^2-1)-1}{q^3-q} + q-1 
 < 2q^i \le q^{i+1}. 
\end{split}
\end{equation}
Hence, we conclude that all these $q^{i}+qt+j$ are pairwise distinct. 
\end{proof}

\subsection{The case with two kinds of cyclotomic cosets}

In this section,  we write $nr=(q^m-1)/s$ for some positive integers $m, s$. 
We also assume that there are exactly two possible values of $l_i$ (recalling $l_i = |C_i|$),
that is,
$\{l_i:i \in Z_{n,r}\}=\{l,m\}$ for some positive integers $l, m$, 
where $l < m$ and $l \mid m$. 

In this case, we can obtain better estimates about the coset leaders $\delta_i^{(l)}$ for some special $i$'s. 

\begin{lemma} \label{lem:deltaNlNm}
  Assume that $\{l_i:i\in Z_{n,r}\}=\{l,m\}$ for some positive integer $l < m$ and $N_l>1$.
  Then,
  $$
  \delta_{\lceil N_l/2 \rceil +1}^{(l)} \geq  \delta_{\lfloor N_l/2 \rfloor+1}^{(l)} >  \frac{qN_m}{2(q-1)} \cdot r, 
  $$
where $N_m = (n-\gcd(n,(q^l-1)/r))/m$.
\end{lemma}

\begin{proof}
  First, since $\{l_i:i\in Z_{n,r}\}=\{l,m\}$, 
  by Lemma~\ref{lem:sumlNl} we have 
$$
N_l=\gcd(n,(q^l-1)/r)/l,
\qquad 
N_m=(n-\gcd(n,(q^l-1)/r))/m. 
$$
So, $N_m <n/m$, and then 
  $$
  \frac{qN_m}{2(q-1)} \cdot r<\frac{q}{q-1} \cdot \frac{nr}{2m}.
  $$
  Hence,  for our purpose it suffices to prove
  \begin{equation}\label{eq:NlNm}
    \delta_{\lfloor N_l/2 \rfloor+1}^{(l)} \ge \frac{q}{q-1} \cdot \frac{nr}{2m}.
  \end{equation}

 Now,  assume that $N_l$ is even.
 Then, by the second lower bound in Lemma~\ref{lem:delta1}  we have 
  \begin{equation*}
    \begin{split}
      \delta_{\lfloor N_l/2 \rfloor +1}^{(l)}
      >&(N_l/2+1-1) \cdot \frac{nr^2}{\gcd(nr,q^l-1)}\\
      =&\frac{\gcd(n,(q^l-1)/r)}{2l} \cdot \frac{nr^2}{\gcd(nr,q^l-1)}\\
      =&\frac{nr}{2l}
      \geq \frac{nr}{m}
      \ge \frac{q}{q-1} \cdot \frac{nr}{2m},
    \end{split}
  \end{equation*}
  where we also use $m \geq 2l$ (since $l \mid m$ and $l < m$).
  So,   \eqref{eq:NlNm} holds in this case.

  Finally,
  assume that $N_l$ is odd and $N_l>1$.
  Considering the case when $q=2$,
  we have $r=1$,
  which implies  $n \in Z_{n,r}$ and  the cyclotomic coset $C_n=\{n\}$. 
  So, in this case 
  we have $l=1$ and $N_l=\gcd(n,1)=1$,
  which contradicts with $N_l>1$.
  Thus, we must have $q \geq 3$. 

In addition, suppose $N_l=q=3$.
  Then, due to $N_l=\gcd(n,(q^l-1)/r)/l$,
  we have $N_l \mid n$.
  So, we have $q \mid n$,
  which contradicts with $\gcd(n,q)=1$.
Hence, we must have either $N_l>3$ or $q>3$. 
Then, we get 
\begin{equation}  \label{eq:qNl}
 \frac{N_l}{N_l-1} \cdot \frac{q}{q-1} \le \frac{3}{3-1}  \cdot \frac{4}{4-1} = 2.
\end{equation}

  By Lemma~\ref{lem:delta1} and noticing that $N_l$ is odd, we obtain 
  \begin{equation*}
    \begin{split}
      \delta_{\lfloor N_l/2 \rfloor +1}^{(l)}
      >&(N_l/2-1/2+1-1) \cdot \frac{nr^2}{\gcd(nr,q^l-1)}\\
      =&\left(\frac{\gcd(n,(q^l-1)/r)}{2l}-\frac{1}{2}\right) \cdot \frac{nr^2}{\gcd(nr,q^l-1)}\\
      =&\left(\frac{1}{2l}-\frac{r}{2\gcd(nr,q^l-1)}\right)  nr \\
      =&\left(\frac{1}{2l}-\frac{1}{2l N_l}\right)  nr
      = \frac{1}{2l}\cdot \frac{N_l-1}{N_l} \cdot nr.
    \end{split}
  \end{equation*}
Then, combining this with \eqref{eq:qNl} and $m \ge 2l$, we have 
  \begin{equation*}
    \begin{split}
      \delta_{\lfloor N_l/2 \rfloor +1}^{(l)}
       > \frac{1}{2l}\cdot  \frac{N_l-1}{N_l} \cdot nr 
       \ge \frac{1}{m}\cdot  \frac{q}{2(q-1)} \cdot nr = \frac{q}{q-1} \cdot \frac{nr}{2m}. 
    \end{split}
  \end{equation*}
This gives  \eqref{eq:NlNm}, and then we complete the proof.
\end{proof}

In Lemma~\ref{lem:deltaNlNm}, if $N_l=1$, then the lower bound there for $\delta_{\lfloor N_l/2 \rfloor+1}^{(l)}$
is not true in general; see Example~\ref{ex:NlNm}.

\begin{example}  \label{ex:NlNm}
Let $q=7$, $n=19$, $r=6$. 
  Then, $m=3$, $s=3$, and the $7$-cyclotomic cosets in $Z_{19,6}$ are: 
$ C_{1}=\{1, 7, 49\}, 
  C_{13}=\{13, 67, 91\}, 
  C_{19}=\{19\}, 
  C_{25}=\{25, 61, 85\}, 
  C_{31}=\{31, 37, 103\},
  C_{43}=\{43, 55, 73\}, 
  C_{79}=\{79, 97, 109\}$.
So, $\{l_i,i\in Z_{19,6}\}=\{1,3\}$, $N_1=1$ and $N_3 = 6$. 
 Then, $\delta_{\lfloor N_1/2 \rfloor+1}^{(1)}=\delta_{1}^{(1)}=19$, but 
  $
  \frac{qN_m}{2(q-1)} \cdot r= \frac{7 \times 6}{12} \times 6=21 > 19.
  $
\end{example}

So, when $N_l=1$, 
we need to impose some  extra condition to achieve such a lower bound.

\begin{lemma}\label{lem:l=1}
  Assume that $\{l_i:i\in Z_{n,r}\}=\{l,m\}$ for some positive integer $l < m$, $N_l=1$ and $2m(q-1)\geq qlr$.
  Then,
  \begin{equation*}
    \delta_{1}^{(l)} > \frac{qN_m}{2(q-1)} \cdot r, 
  \end{equation*}
where $N_m = (n-\gcd(n,(q^l-1)/r))/m$.
\end{lemma}

\begin{proof}
As in \eqref{eq:NlNm} and noticing $N_l=1$, for our purpose it suffices to prove 
\begin{equation}  \label{eq:N1Nm}
  \delta_{1}^{(l)} \ge \frac{q}{q-1} \cdot \frac{nr}{2m}. 
\end{equation}

  First,  assume $q=2$,  then we have $r=1$,
  which implies $n \in Z_{n,r}$ and the cyclotomic coset $C_n=\{n\}$. 
  So,  we have $l=1$ and $N_l=\gcd(n,1)=1$.
  Hence, 
  $$
  \delta_{1}^{(l)}=n>\frac{q}{q-1} \cdot \frac{nr}{2m},
  $$
  which implies \eqref{eq:N1Nm} when $q=2$.

  Now,  assume $q>2$.
  By Lemma~\ref{lem:sumlNl},
  we have $N_l=\gcd(n,(q^l-1)/r) / l$. 
Then, noticing $N_l=1$, we obtain 
  $\gcd(n,(q^l-1)/r)=l$. 
  Combining this with the condition $2m(q-1)\geq qlr$ and using the first lower bound in Lemma~\ref{lem:delta1},
  we have
  \begin{equation*}
    \begin{split}
      \delta_{1}^{(l)}\geq \frac{nr}{\gcd(nr,q^l-1)}
      =\frac{n}{l}
      =\frac{2m}{l} \cdot \frac{n}{2m}
      \geq\frac{qlr}{l(q-1)} \cdot  \frac{n}{2m}
      =\frac{q}{q-1} \cdot \frac{nr}{2m}.
    \end{split}
    \end{equation*}
  Thus, 
  \eqref{eq:N1Nm} also holds in this case.
  This completes the proof.
\end{proof}

In Lemma~\ref{lem:l=1},
when $r=1,2 \text{ or } 3$,
the extra condition 
$2m(q-1)\geq qlr$ in fact always holds.

\begin{corollary}   \label{cor:delta1}
  Assume that $\{l_i:i\in Z_{n,r}\}=\{l,m\}$ for some positive integer $l < m$, $N_l=1$ and $r \le 3$.
   Then,
  \begin{equation*}
    \delta_{1}^{(l)} > \frac{qN_m}{2(q-1)} \cdot r, 
  \end{equation*}
where $N_m = (n-\gcd(n,(q^l-1)/r))/m$.
\end{corollary}

\begin{proof}
  Since $r \mid q-1$,
  we have $r \leq q-1$.
  So,
  \begin{equation}\label{eq:q=r+1}
    \frac{q}{q-1} \cdot r \leq \frac{r+1}{r} \cdot r=r+1.
  \end{equation}
  Noticing $m/l \geq 2$ and $r \leq 3$,
  we have 
  \begin{equation}\label{eq:2m/l>r+1}
    2m/l \geq 4 \geq r+1.
  \end{equation}
  Combining \eqref{eq:q=r+1} with \eqref{eq:2m/l>r+1},
  we obtain 
  $2m/l \geq \frac{q}{q-1} \cdot r,$
  that is,
  $
  2m(q-1) \geq q lr.
  $
  Thus,
  by Lemma~\ref{lem:l=1},
  we get the desired result. 
\end{proof}

When $l=1$ in Lemma~\ref{lem:deltaNlNm},  we can do a little better. 
One can see this by combining the following lemma with Lemma~\ref{lem:deltaNm}. 

\begin{lemma}  \label{lem:deltaN1Np}
  Assume $\{l_i:i\in Z_{n,r}\}=\{1,m\}$ with $m \ge 2$.
 Then, if $N_1 \ge 2$, we have 
$$\delta_{\lceil N_1/2 \rceil+1}^{(1)}> \delta_{\lceil N_m/2 \rceil}^{(m)};$$
and if $N_1$ is even or $r \leq 2$,  we have
$$\delta_{\lfloor N_1/2 \rfloor +1}^{(1)}> \delta_{\lfloor N_m/2 \rfloor}^{(m)}.$$
\end{lemma}

\begin{proof}
 Using Lemma~\ref{lem:sumlNl} and noticing  $\{l_i:i\in Z_{n,r}\}=\{1,m\}$,
  we have
  \begin{equation}\label{eq:Nl1p}
     N_1 = \gcd(n,(q -1)/r),  \quad
     N_m = (n-\gcd(n,(q -1)/r))/m.
  \end{equation}

First, we estimate $\delta_{\lceil N_1/2 \rceil+1}^{(1)}$.
Combining \eqref{eq:Nl1p} with Lemma~\ref{lem:delta1},
  we have
  \begin{equation}\label{eq:>nr/21}
    \begin{split}
    \delta_{\lceil N_1/2 \rceil+1}^{(1)} 
    \ge & (\lceil N_1/2 \rceil+1-1) \cdot \frac{nr^2}{\gcd(nr,q-1)} + 1\\
    \geq &\frac{\gcd(n,(q-1)/r)}{2} \cdot \frac{nr^2}{\gcd(nr,q-1)} + 1 \\
    =& \frac{nr}{2} +1.
    \end{split}
  \end{equation}  
So, for our purpose we only need to prove 
\begin{equation}  \label{eq:ceilNp}
\delta_{\lceil N_m/2 \rceil}^{(m)} <  \frac{nr}{2} +1. 
\end{equation}

By contradiction, we suppose $\delta_{\lceil N_m/2 \rceil}^{(m)} \ge  \frac{nr}{2} +1$. 
Then, combining \eqref{eq:Nl1p} with \eqref{eq:>nr/21}, 
we deduce that in the set $Z_{n,r}$, the total number of elements in all the $q$-cyclotomic cosets with coset leader less 
than $\frac{nr}{2} +1$ is at most 
  \begin{equation}\label{eq:N1Np}
    \begin{split}
       \lceil N_1/2 \rceil + m \left( \lceil N_m/2 \rceil -1 \right) 
      & \le \frac{\gcd(n,(q-1)/r)}{2} + \frac{1}{2} + m \cdot \frac{n-\gcd(n,(q-1)/r)}{2m}- \frac{m}{2} \\
  & =\frac{n}{2}- \frac{m-1}{2}.
    \end{split}
  \end{equation}
However, in the set $Z_{n,r}$ the number of elements less than $\frac{nr}{2}+1$ is at least 
$
\lceil n/2 \rceil \ge n/2,
$
which contradicts with \eqref{eq:N1Np} (noticing $m \ge 2$). 
Hence, \eqref{eq:ceilNp} is true, and thus we get  
$\delta_{\lceil N_1/2 \rceil+1}^{(1)}> \delta_{\lceil N_m/2 \rceil}^{(m)}.$

Now, it remains to consider the case when  $N_1$ is odd and $r \leq 2$. 
Since $N_1$ is odd,  
  by Lemma~\ref{lem:delta1} we obtain
  \begin{equation}  \label{eq:>nr/23}
    \begin{split}
      \delta_{\lfloor N_1/2 \rfloor+1}^{(1)}
      & \geq (\lfloor N_1/2 \rfloor+1-1) \cdot \frac{nr^2}{\gcd(nr,q-1)}+ \frac{nr}{\gcd(nr,q-1)}\\
      & = \left(\frac{\gcd(n,(q-1)/r)}{2}- \frac{1}{2}\right) \cdot \frac{nr^2}{\gcd(nr,q-1)}+ \frac{nr}{\gcd(nr,q-1)}\\
      & = \frac{nr}{2} + \frac{nr(2-r)}{2\gcd(nr,q-1)} \ge \frac{nr}{2},
    \end{split}
  \end{equation}
  where the last inequality follows from the assumption $r \leq 2$.
  Hence, when  $N_1$ is odd and $r \leq 2$, 
  for our purpose we only need to prove 
\begin{equation}\label{eq:>nr/24}
   \delta_{\lfloor N_m/2 \rfloor}^{(m)}  <    \frac{nr}{2}.
  \end{equation}

By contradiction, suppose $\delta_{\lfloor N_m/2 \rfloor}^{(m)} \ge nr/2$. 
Then, combining \eqref{eq:Nl1p} with \eqref{eq:>nr/23}, 
we deduce that in the set $Z_{n,r}$, the total number of elements in the $q$-cyclotomic cosets with coset leader less 
than $nr/2$ is at most 
  \begin{equation}\label{eq:N1Np2}
    \begin{split}
       \lfloor N_1/2 \rfloor + m \left( \lfloor N_m/2 \rfloor -1 \right) 
      & \le \frac{\gcd(n,(q-1)/r)}{2} + m \cdot \frac{n-\gcd(n,(q-1)/r)}{2m}- m \\
  & =\frac{n}{2}- m.
    \end{split}
  \end{equation}
However, in the set $Z_{n,r}$ the number of elements less than $nr/2$ is at least 
$
\lfloor n/2 \rfloor - 1  \ge \frac{n}{2}- \frac{3}{2},
$
which contradicts with \eqref{eq:N1Np2} (noticing $m \ge 2$). 
Therefore, \eqref{eq:>nr/24} is true, and so, we obtain 
$
\delta_{\lfloor N_1/2 \rfloor +1}^{(1)}> \delta_{\lfloor N_m/2 \rfloor}^{(m)}.
$
This completes the proof. 
\end{proof}

We remark that if $N_1$ is odd and $r>2$,
then
$\delta_{\lfloor N_1/2 \rfloor +1}^{(1)}$ may be less than
$\delta_{\lfloor N_m/2 \rfloor}^{(m)}$; see Example~\ref{ex:deltaN1Nm}.

\begin{example}\label{ex:deltaN1Nm}
Choosing $q=19$, $n=127$ and $r=18$,
  we have $nr=(q^m-1)/s$ with  $m=3$, $s=3$, 
  $\{l_i,i\in Z_{127,18}\}=\{1,3\}$,
  $N_1=1$ and $N_3=42$. 
Moreover, $\delta_1^{(1)}=127$ and $\delta_{21}^{(3)}=451$. 
So, 
$
\delta_{\lfloor N_1/2 \rfloor +1}^{(1)}=\delta_1^{(1)}<\delta_{21}^{(3)}=\delta_{\lfloor N_3/2 \rfloor}^{(m)}.
$
In addition, $\delta_1^{(1)}$ is also less than the 21-th element 361 in the set $Z_{127,18}$, 
and so, $\delta_1^{(1)}$ is not in the first 21 cyclotomic cosets of size 3 in $Z_{127,18}$.
\end{example}

\subsection{Cyclotomic values}

Recall that for any positive integer $m$, the $m$-th cyclotomic polynomial is defined to be 
$$
\Phi_m(x)= \prod_{
  \substack{j=1\\ \gcd(j,m)=1}}^{m} (x-\zeta_m^j) \in \Z[x], 
$$
where $\zeta_m$ is a primitive $m$-th root of unity.
Besides, $\Phi_m(x)$ can be expressed as 
\begin{equation}\label{eq:Phimx}
  \Phi_m(x)=\prod_{d \mid m} (x^d-1)^{\mu(m/d)}.
\end{equation}

The following result is in \cite[Proposition 2.8]{Washington} about $\Phi_m(1)$.

\begin{lemma}[\cite{Washington}] \label{lem:Phi(1)}
  If $m$ has at least two distinct prime factors, then $\Phi_m(1)=1$.  
\end{lemma}

Next, we determine $\gcd(\Phi_m(q),q-1)$. 

\begin{lemma}\label{lem:gcdPhiq-1}
We have
  \begin{align*}
    \gcd(\Phi_m(q),q-1)
    =\begin{cases}
      1,& m \text{ has at least two distinct prime factors},\\
      \gcd(q-1,p), & m \text{ is a power of a prime } p.
    \end{cases}
  \end{align*}
\end{lemma}

\begin{proof}
  First, we assume that $m$ has at least two distinct prime factors.
  By contradiction, suppose $\gcd(\Phi_m(q),q-1)>1$. 
  Then, there exists a prime $p$ such that $p \mid \gcd(\Phi_m(q),q-1)$.
  However, by Lemma~\ref{lem:Phi(1)} and noticing $p \mid q-1$, we have
  $$
  \Phi_m(q) \equiv \Phi_m(1) \equiv 1 \pmod{p},
  $$
  which contradicts with $p \mid \Phi_m(q)$.
  Thus, $\gcd(\Phi_m(q),q-1)=1$.

  Now, we assume $m=p^{j+1}$ for some non-negative integer $j$.
  By \eqref{eq:Phimx},  we obtain 
  \begin{align*}
      \Phi_m(q)  =\frac{q^{p^{j+1}}-1}{q^{p^j}-1} 
    &=q^{p^{j}(p-1)}+q^{p^j(p-2)}+\cdots+q^{p^j}+1\\ 
    &\equiv 1+1+\cdots+1+1 \equiv p \pmod{q-1}. 
  \end{align*}
  So, we have
  $
  \gcd(\Phi_m(q),q-1)=\gcd(q-1,p).
  $
  This proves the second part of the lemma.
\end{proof}

\subsection{The BCH bound}
The BCH lower bound is an important bound for the minimun distance of BCH codes,
which can be extended for constacyclic codes; see \cite[Theorem 2.2]{KS}.

\begin{lemma}[\cite{KS}] \label{lem:bchbound}
  If the defining set $T$ of a $\lambda$-constacyclic code $\C$ over $\F_q$ includes $1, 1+r,1+2r,\dots,1+dr$,
  then the minimun distance $d(\C) \geq d+2$. 
\end{lemma}

\subsection{Equivalence between constacyclic codes and cyclic codes}

It is well-known that a $\lambda$-constacyclic code over $\F_q$ is 
monomially equivalent to a cyclic code of the same length over the field $\F_q(\alpha)$, 
where $\alpha$ is an $n$-th root of $\lambda$; see \cite[Theorem 2.1]{ASR}. 
The following lemma is the special case when $\alpha \in \F_q$; see  \cite[Corollary 2.1 and Lemma 2.1]{ASR}. 
Recall that $\lambda \in \F_q^*$ and $r$ is the multiplicative order of $\lambda$. 

\begin{lemma}[\cite{ASR}] \label{lem:equivalence}
  If $\gcd(n,q-1) \mid \frac{q-1}{r}$, then a $\lambda$-constacyclic code of length $n$ over
$\F_q$ is monomially equivalent to a cyclic code of length $n$ over $\F_q$.  
\end{lemma}

From Lemma \ref{lem:equivalence}, we get the following corollary. 

\begin{corollary} \label{cor:equivalence}
  If $\gcd(r, \frac{nr}{\gcd(nr,q-1)})=1$, then a $\lambda$-constacyclic code of length $n$ over
$\F_q$ is monomially equivalent to a cyclic code of length $n$ over $\F_q$.  
\end{corollary}

\begin{proof}
Since $r \mid q-1$, we have  $\gcd(r, \frac{nr}{\gcd(nr,q-1)})= \gcd(r, \frac{n}{\gcd(n,(q-1)/r)})$. 
So, by assumption we have $\gcd(r, \frac{n}{\gcd(n,(q-1)/r)})=1$. 
Then, we get 
\begin{align*}
\gcd(n,q-1) = \gcd(n,r \cdot (q-1)/r)
& = \gcd(n, \frac{q-1}{r}) \cdot 
\gcd(\frac{n}{ \gcd(n, \frac{q-1}{r})}, r ) \\
& = \gcd(n, \frac{q-1}{r}),
\end{align*}
which implies $\gcd(n,q-1) \mid \frac{q-1}{r}$. 
Hence, the desired result follows from Lemma~\ref{lem:equivalence}. 
\end{proof}

We remark that the condition $\gcd(r, \frac{nr}{\gcd(nr,q-1)})=1$ 
in Corollary~\ref{cor:equivalence} is equivalent to that $N_1 >0$ 
(that is, there exists a $q$-cyclotomic coset with one element in $Z_{n,r}$ modulo $nr$); 
see Lemma~\ref{lem:sumlNl}.

\section{Some general constructions for constacyclic codes} \label{sec:cons}

In this section, motivated by \cite{SLD} and generalizing the constructions in \cite{CW, LGLS, LWS, SL, SLD},  
to construct a $q$-ary constacyclic code, 
we first classify the $q$-cyclotomic cosets in $Z_{n,r}$ according to their sizes, 
and then we take the defining set $T$ as the union of the first half of cyclotomic cosets in each class 
(sorted in ascending order by their coset leaders). 

Recall that $q$ is a prime power,
$n$ is a positive integer with $\gcd(n,q)=1$, and $r \mid q-1$. 

For simplicity, we still denote $N_l = N_l^{(q,n,r)}$ and $\delta_i^{(l)} = \delta_i^{(q,n,r,l)}$. 
By definition,  $N_l$ is the number of $q$-cyclotomic cosets with $l$ elements in the set $Z_{n,r}$, 
and $\delta_i^{(l)}$ is the $i$-th coset leader (in ascending order) in those $q$-cyclotomic cosets with $l$ elements in  $Z_{n,r}$.  

Recall that for any $i \in Z_{n,r}$, $l_i = |C_i|$, which is the size of the $q$-cyclotomic coset 
generated by $i$ modulo $nr$ (we denote $C_i = C_i^{(q,nr)}$). 

\begin{theorem}\label{thm:NlNm}
  Write $nr=(q^m-1)/s$, and 
  assume $\{l_i:i \in Z_{n,r}\}=\{l,m\}$ for some positive integer $l < m$.
 Then, an infinite family of $q$-ary $[n,n-l \lceil N_l/2 \rceil - m \lceil N_m/2 \rceil,\ge \lfloor \frac{q N_m}{2(q-1)} \rfloor]$ constacyclic  codes is constructed, where 
  $$
  N_l=\gcd(n,(q^l-1)/r)/l \quad \textrm{and} \quad N_m=(n-\gcd(n,(q^l-1)/r))/m. 
  $$ 
 Besides, if
  \begin{itemize}
    \item [{\rm (1)}] $N_l>1$, or
    \item [{\rm (2)}] $N_l=1$ and $2m(q-1) \geq qlr$, or 
    \item [{\rm (3)}] $r \leq 3$, 
  \end{itemize} 
  then an infinite family of $q$-ary $[n,n-l \lfloor N_l/2 \rfloor - m \lfloor N_m/2 \rfloor,\ge \lfloor \frac{qN_m}{2(q-1)} \rfloor]$  constacyclic codes is constructed.
\end{theorem}

\begin{proof}
  Since  $\{l_i:i \in Z_{n,r}\}=\{l,m\}$ with $l < m$, by Lemma~\ref{lem:sumlNl} 
  we have 
  $$
  N_l=\gcd(n,(q^l-1)/r)/l \quad \textrm{and} \quad N_m=(n-\gcd(n,(q^l-1)/r))/m.
  $$

First, take the defining set $T$ as the union of the first $\lceil N_l/2 \rceil$ $q$-cyclotomic cosets 
 contained in $Z_{n,r}$ with $l$ elements
  and the first $\lceil N_m/2 \rceil$ $q$-cyclotomic cosets 
 contained in $Z_{n,r}$ with $m$ elements.
  So, there are $l \lceil N_l/2 \rceil + m \lceil N_m/2 \rceil$ integers in $T$,
  which implies that the dimension of the corresponding constacyclic code $\C$ is 
$
n-l \lceil N_l/2 \rceil - m \lceil N_m/2 \rceil.
$
  
  Next,   we prove the desired lower bound for the minimun distance $d(\C)$. 
Since $d(\C) \ge 2$, we can assume the desired lower bound $\lfloor \frac{q N_m}{2(q-1)} \rfloor \ge 2$ without loss of generality.  
When $N_l > 1$, 
  by Lemmas~\ref{lem:deltaNm} and \ref{lem:deltaNlNm},
  we have 
  $$
  \delta_{\lceil N_l/2 \rceil+1}^{(l)} > \frac{q N_m}{2(q-1)} \cdot r>1+\left(\frac{q N_m}{2(q-1)}-2\right)r,
  $$ 
  $$
  \delta_{\lceil N_m/2 \rceil+1}^{(m)} > 1+\left(\frac{q N_m}{2(q-1)}-2\right)r. 
  $$
By construction, this means that for any non-negative integer $i \le \frac{q N_m}{2(q-1)}-2$,
 we have $1+ir \in T$. 
So, by Lemma~\ref{lem:bchbound} we get 
$
d(\C) \ge \left\lfloor \frac{qN_m}{2(q-1)} \right\rfloor.
$
Similarly, the case when $N_l=1$ follows from Lemmas~\ref{lem:deltaNm} and \ref{lem:bchbound}. 
So, we prove the first part of the theorem.

  Now, take the defining set $T$ as the union of the first $\lfloor N_l/2 \rfloor$ $q$-cyclotomic cosets contained in $Z_{n,r}$ with $l$ elements
  and the first $\lfloor N_m/2 \rfloor$ $q$-cyclotomic cosets contained in $Z_{n,r}$ with $m$ elements.
Then, under one of the conditions (1), (2) and (3) and using 
Lemmas~\ref{lem:deltaNm}, \ref{lem:deltaNlNm} and~\ref{lem:l=1} and Corollary~\ref{cor:delta1},
  an infinite family of constacyclic $[n,n-l \lfloor N_l/2 \rfloor - m \lfloor N_m/2 \rfloor,\ge \lfloor \frac{qN_m}{2(q-1)} \rfloor]$ codes is constructed.
\end{proof}

\begin{remark}   \label{rem:cyclic}
In Theorem~\ref{thm:NlNm}, if moreover $l=1$, then $N_1 > 0$; 
and so, by Lemma~\ref{lem:sumlNl} we have  $\gcd(r, \frac{nr}{\gcd(nr,q-1)})=1$. 
Then, by Corollary~\ref{cor:equivalence} we know that the constacyclic codes 
in Theorem~\ref{thm:NlNm} are monomially equivalent to some cyclic codes over $\F_q$.
\end{remark}

We also remark that Theorem~\ref{thm:NlNm} includes the codes constructed in \cite[Theorems 10, 15 and 16]{CW} and also improves the lower bounds on minimum distances there. 

We notice that  in Theorem~\ref{thm:NlNm}, the quantity $N_m$ satisfies 
$
N_m \ge \frac{n}{2m} = \frac{n}{2\log_q(nrs+1)}. 
$
This means that when $s \le n^{10}$, the minimum distance $d(\C)$ is greater than $cn/\log_q n$ 
for some positive constant $c$ depending only on $q$.  

We emphasize that the codes in Theorem~\ref{thm:NlNm} are not always BCH constacyclic codes; 
see the following example for the case when $r=1$. 

\begin{example}   \label{ex:not-BCH} 
    Let $q=5$, $n=26$, and $r=1$. Then, 
  the $5$-cyclotomic cosets modulo $26$ are: 
  $
  C_{1}=\{1, 5, 21, 25\}, 
  C_{2}=\{2, 10, 16, 24\}, 
  C_{3}=\{3, 11, 15, 23\}, 
  C_{4}=\{4, 6, 20, 22\}, 
   C_{7}=\{7, 9, 17, 19\}, 
  C_{8}=\{8, 12, 14, 18\},
  C_{13}=\{13\}, 
  C_{26}=\{26\}.
  $
  Then, $N_1 = 2, N_4 =6$. 
  So, in Theorem~\ref{thm:NlNm}, 
  the defining set is $C_1 \cup C_2 \cup C_3 \cup  C_{13}$,
   and the relevant code is not a BCH code. 
\end{example}

Next, by extending the strategy in \cite[Theorem 10]{SLD}, we establish a better lower bound for the minimum distance in some special cases. 

\begin{theorem}  \label{thm:NlNm2}
  Write $nr=(q^m-1)/s$, and 
  assume $\{l_i:i \in Z_{n,r}\}=\{l,m\}$ for some positive integer $l<m$. 
 Then, if either $l=1$ or $N_l = 1$,  for the first family of $q$-ary constacyclic codes in Theorem~\ref{thm:NlNm}, the lower bound on their minimum distances can be 
$$ 
\left\lceil \frac{q \lceil N_m/2 \rceil}{q-1} \right\rceil + \Lambda_m,
$$
where $N_m$ is given in Theorem~\ref{thm:NlNm}, 
$
 \Lambda_m = \left\lceil \frac{q|S_m|}{q-1} \right\rceil - 1, 
$
 and 
\begin{equation*}
\begin{split}
 S_{m} = \Big\{(i,j,t):\ & \frac{m+1}{2} \le i \le \lfloor m - \log_q(s+1) \rfloor, \\ 
& 1 \le j \le q-1, \
 0 \le t <  \frac{q^{i}-j(q^{m-i}-1)-1}{q^{m-i+1}-q},  \\
& q^i+qt+j \equiv 1 \pmod{r}, \\
& q^i+qt+j \le 1 + \left(\lceil q\lceil N_m/2 \rceil/(q-1) \rceil -2 \right) r \Big\}; 
\end{split}
\end{equation*} 
and in addition, if $l=1$ and $r \le 2$,   for the second family of $q$-ary constacyclic codes in Theorem~\ref{thm:NlNm}, the lower bound on their minimum distances can be 
$$ 
\left\lceil \frac{q \lfloor N_m/2 \rfloor}{q-1} \right\rceil  + \Lambda_m^\prime, 
$$
where 
$
 \Lambda_m^\prime = \left\lceil \frac{q|S_m^\prime|}{q-1} \right\rceil - 1, 
$
and the set $S_m^\prime$ is obtained from $S_m$ by replacing $\lceil N_m/2 \rceil$ with $\lfloor N_m/2 \rfloor$.
\end{theorem}

\begin{proof}
Recall that for the first family of constacyclic codes in Theorem~\ref{thm:NlNm}, 
their defining set $T$ is the union of the first $\lceil N_1/2 \rceil$ cyclotomic cosets contained in $Z_{n,r}$ with one element
  and the first $\lceil N_m/2 \rceil$ cyclotomic cosets contained in $Z_{n,r}$ with $m$ elements. 

By construction and by the first part of Lemma~\ref{lem:deltaN1Np}, we know that when $l=1$ and $N_1 \ge 2$, all the cyclotomic cosets of $Z_{n,r}$  with coset leader at most $\delta_{\lceil N_m/2 \rceil}^{(m)}$ are in $T$. 
When $N_l=1$, clearly this fact also holds.

Moreover, by Lemma~\ref{lem:delta1} we get 
\begin{align*}
  \delta^{(m)}_{\lceil N_m/2 \rceil}  
 \ge \frac{nr(1 + (\lceil q\lceil N_m/2 \rceil/(q-1) \rceil -2) r)}{\gcd(nr,q^m-1)} 
  = 1 + (\lceil q\lceil N_m/2 \rceil/(q-1) \rceil -2) r.
\end{align*}
Notice that the above lower bound is achieved only when 
 all the elements not divisible by $q$ and at most $1 + \left(\lceil q\lceil N_m/2 \rceil/(q-1) \rceil -2 \right) r$ in 
$Z_{n,r}$ are coset leaders (because in proving Lemma~\ref{lem:delta1}, we apply Lemma~\ref{lem:ar}). 
However, this does not always hold. 
By Lemma~\ref{lem:NotLeader}, we know that there are at least $|S_m|$ such elements which are not coset leaders, where the set $S_m$ has been defined in the theorem. 
Hence, combining this with Lemma~\ref{lem:ar}, we obtain 
\begin{equation}  \label{eq:deltaNp2}
\begin{split}
  \delta^{(m)}_{\lceil N_m/2 \rceil}  & \ge  1 + \left(\lceil q\lceil N_m/2 \rceil/(q-1) \rceil - 1 \right) r 
 + (\lceil q|S_m|/(q-1)\rceil -2 )r \\ 
& \ge 1 + \left(\lceil q\lceil N_m/2 \rceil/(q-1) \rceil + \lceil q|S_m|/(q-1)\rceil - 3 \right) r.
\end{split}
\end{equation}

Hence, using \eqref{eq:deltaNp2} and applying Lemma~\ref{lem:bchbound}, we get the desired lower bound 
for the minimum distances of  the first family of constacyclic codes in Theorem~\ref{thm:NlNm}. 

Finally, when $l=1$ and $r \le 2$, using the second part of Lemma~\ref{lem:deltaN1Np} and applying similar arguments as the above, we get the desired lower bound 
for the minimum distances of  the second family of constacyclic codes in Theorem~\ref{thm:NlNm}.
\end{proof}

We remark that the sizes of the sets $S_m$ and $S_m^\prime$ in Theorem~\ref{thm:NlNm2} can be easily computed by computer. 
We give an estimate for $|S_m|$ below (the method also works for estimating $|S_m^{\prime}|$), although it is not always good according to numerical data. 
Let 
\begin{equation*}
e(m) = \min\{\lfloor m - \log_q(s+1) \rfloor, \ \lfloor \log_q(1 + \left(\lceil q\lceil N_m/2 \rceil/(q-1) \rceil -2 \right) r)  \rfloor\}.
\end{equation*}
For any $i$ with $(m+1)/2 \le i \le e(m)-1$ and for any $j,t$ described in $S_m$, 
by \eqref{eq:qi+1} and noticing the choice of $e(m)$, we have 
$$
q^i + qt+j < q^{i+1} \le q^{e(m)} \le 1 + \left(\lceil q\lceil N_m/2 \rceil/(q-1) \rceil -2 \right) r. 
$$
Therefore, using the second part of Lemma~\ref{lem:NotLeader}, we obtain 
\begin{equation*}  
|S_m| \ge  \sum_{i=(m+1)/2}^{e(m)-1} \sum_{j=1}^{q-1} \left\lfloor \frac{q^{i}-j(q^{m-i}-1)-1}{r(q^{m-i+1}-q)} \right\rfloor, 
\end{equation*}
where $r$ appears in the denomenator due to the condition $q^i+qt+j \equiv 1 \pmod{r}$.

The data in Table~\ref{tab:lower} suggest that the improvement of the lower bound in 
Theorem~\ref{thm:NlNm2} is indeed meaningful (compared to Theorem~\ref{thm:NlNm}).

\begin{table}[h]
 \caption{The lower bound in Theorem \ref{thm:NlNm2}}
 \label{tab:lower}
\vspace{3mm}
\centering
\begin{tabular}{ccccccc}
    \hline 
    $q$&$m$&$s$&$r$ & $\left\lfloor \frac{q N_m}{2(q-1)} \right\rfloor$ & $\left\lceil \frac{q \lceil N_m/2 \rceil}{q-1} \right\rceil + \Lambda_m$ &  Bose distance \\[2mm]
     \hline
  3&5&1&1&36&38&41\\
  3&11&1&2&6039&6179&7266\\
  4&7&1&1&1560&1631&1835\\
  4&7&1&3&520&543&613\\
  5&7&2&1&3487&3622&4157\\
  5&7&2&2&1743&1811&2080\\
    \hline
\end{tabular}
\end{table}

\section{Constacyclic codes with $n=\frac{q^p -1}{rs}$}  \label{sec:qp}

In this section, we assume $nr = \frac{q^p - 1}{s}$ for some prime $p$ and some positive integer $s$. 
Recall that $r \mid q-1$. 

\begin{theorem}  \label{thm:qp}
  Assume that $nr=(q^p-1)/s$ for some prime $p$, and $\gcd(r,nr/\gcd(nr,q-1))=1$. 
 Then, an infinite family of $q$-ary $[n,n- \lceil N_1/2 \rceil - p \lceil N_p/2 \rceil,\ge \left\lceil \frac{q \lceil N_p/2 \rceil}{q-1} \right\rceil + \Lambda_p]$ constacyclic codes is constructed, where 
  $$
  N_1=\gcd(n,(q-1)/r), \qquad N_p=(n-\gcd(n,(q-1)/r))/p, 
  $$ 
and $\Lambda_p$ has been defined in Theorem~\ref{thm:NlNm2}.
 Moreover, if
  \begin{itemize}
    \item [{\rm (1)}] $N_1 > 1$, or
    \item [{\rm (2)}] $N_1=1$ and $2p(q-1) \geq qr$, or 
    \item [{\rm (3)}] $r \leq 3$, 
  \end{itemize} 
  then an infinite family of $q$-ary $[n,n- \lfloor N_1/2 \rfloor - p \lfloor N_p/2 \rfloor,\ge \lfloor \frac{qN_p}{2(q-1)} \rfloor]$ constacyclic  codes is constructed.
\end{theorem}

\begin{proof}
Since $nr = (q^p -1)/s$ and $p$ is a prime, 
  by Lemma~\ref{lem:liord} we have that 
  for any $i\in Z_{n,r}$, $l_i=1$ or $p$, where $l_i = |C_i|$.  
  Combining this with Lemma~\ref{lem:sumlNl} (noticing $\gcd(r,nr/\gcd(nr,q-1))=1$),
  we have
  $$
N_1= \gcd(n,(q-1)/r),   \quad
N_p = (n-\gcd(n,(q-1)/r))/p. 
$$
  Then,  the desired results follow directly from Theorems~\ref{thm:NlNm} and \ref{thm:NlNm2}. 
\end{proof}

We remark that in Theorem~\ref{thm:qp}, since $N_1 >0$, 
by Remark~\ref{rem:cyclic} we know that the constacyclic codes 
in Theorem~\ref{thm:qp} are monomially equivalent to some cyclic codes over $\F_q$.

We notice that some codes in Theorem~\ref{thm:qp} are better than BCH constacyclic codes; see
Example~\ref{ex:qp}. 

\begin{example}   \label{ex:qp}
Let $q=7$, $p=3$, $s=9$, $r=1$, and $n=38$. 
  Then, the $7$-cyclotomic cosets in $Z_{38, 1}$ are: 
  $C_{1}=\{1, 7, 11\}$, 
  $C_{2}=\{2, 14, 22\}$, 
  $C_{3}=\{3, 21, 33\}$, 
  $C_{4}=\{4, 6, 28\}$, 
  $C_{5}=\{5, 17, 35\}$, 
  $C_{8}=\{8, 12, 18\}$,
  $C_{9}=\{9, 23, 25\}$, 
  $C_{10}=\{10, 32, 34\}$, 
$C_{13} =\{13, 15, 29\}$, 
$C_{16} =\{16, 24, 36\}$,
$C_{19}=\{19\}$,
$C_{20}=\{20, 26, 30\}$, 
$C_{27} =\{27, 31, 37\}$, 
$C_{38}=\{38\}$.
So, $\{l_i,i\in Z_{38,1}\}=\{1,3\}$, $N_1=2$ and $N_3 = 12$. 
From Theorem~\ref{thm:qp}, we take the defining set as the set 
$C_1 \cup C_2 \cup C_3 \cup C_4 \cup C_5 \cup C_8 \cup C_{19}$ and obtain a $7$-ary $[38, 19, 13]$ cyclic code, which is not a BCH code and has the best-known parameters \cite{Grassl}. 
In this case, there are several 7-ary BCH codes of length 38 and dimension 19, 
but their minimum distances are at most 12. 
\end{example}

 In Table~\ref{tab:qp}, we list some good codes constructed in Theorem~\ref{thm:qp}. 
In the table, those codes marked with ``*" in the parameters are constructed by using the ceiling function 
(that is, from the first family of constacyclic codes in the theorem), 
and the others  are constructed by using the floor function 
(that is, from the second family of constacyclic codes in the theorem). 
In addition, the column ``Lower bound" records the best lower bound in this paper for the minimum distance of the relevant code (here, actually from Theorems~\ref{thm:NlNm}, \ref{thm:NlNm2} and \ref{thm:qp}). 
We follow these rules in all the other tables of constacyclic codes in this paper.

\begin{longtable}{cccccccc}   
  \caption{Codes in Theorem \ref{thm:qp}}
  \label{tab:qp}\\
    \hline
    $q$&$p$&$s$&$r$ & Lower bound & Bose distance &Parameter&Optimality\\
    \hline
    \endfirsthead
    \hline
      $q$&$p$&$s$&$r$ & Lower bound & Bose distance &Parameter&Optimality\\
    \hline
    \endhead
    \hline
  \endfoot
    2&11&89&1&2&6&$[23,11,8]^*$&\text{optimal}\\
    2&11&89&1&2&5&$[23,12,7]$&\text{optimal}\\
    3&5&11&1&3&4&$[22,11,7]^*$&\text{almost-optimal}\\
    3&11&3851&2&1&6&$[23,11,9]^*$&\text{optimal}\\
    3&11&3851&1&1&5&$[23,12,8]$&\text{optimal}\\
  7&3&18&1&4&6&$[19,9,8]^*$&\text{almost-optimal}\\  
  7&3&6&3&3&6&$[19,10,7]$&\text{almost-optimal}\\  
  7&3&9&1&7&9&$[38,19,13]$&best-known\\ 
  7&7&9466&3&3&4&$[29,14,12]^*$&best-known\\ 
  7&7&9466&3&2&4&$[29,15,11]$&best-known\\
\end{longtable}

\section{Constacyclic codes with $n=\Phi_{p^{b}}(q)$}  \label{sec:pbq}

Recall that $q$ is a prime power, and $r \mid q-1$.
Let $p$ be a prime and $b$ be a positive integer.
In this section, we construct some infinite families of constacyclic codes with length 
$$
n=\Phi_{p^{b}}(q)=\frac{q^{p^{b}}-1}{q^{p^{b-1}}-1}.
$$
Later on, we will see that  in this case the set  $\{l_i: \, i\in Z_{n,r}\}$, 
where $l_i = |C_i|$, always has exactly two elements. 
So, we can use Theorem~\ref{thm:NlNm} to construct constacyclic codes having good lower bounds for their minimum distances.

First, we consider the case when $p \nmid q-1$. 

\begin{theorem}\label{thm:pb1}
  Assume $p \nmid q-1$ and $n=\Phi_{p^{b}}(q)$.
  Then,  an infinite family of $q$-ary  $[n,n-1-p^{b}\lceil(n-1)/(2p^{b}) \rceil,\ge \left\lceil \frac{q \lceil (n-1)/(2p^b) \rceil}{q-1} \right\rceil + \Lambda_{p^b}]$ constacyclic codes is constructed, 
where $\Lambda_{p^b}$ has been defined in Theorem~\ref{thm:NlNm2}.  
  Moreover, if either $r \le 3$ or $2p^{b}(q-1)\ge qr$,
  an infinite family of $q$-ary $[n,n-p^{b} \lfloor(n-1)/(2p^{b}) \rfloor,\ge \lfloor \frac{q(n-1)}{2p^{b}(q-1)} \rfloor]$ constacyclic  codes is constructed.
\end{theorem}

\begin{proof}
  First, we want to determine the set $\{l_i: \, i\in Z_{n,r}\}$, where $l_i = |C_i|$. 
Since $n=\Phi_{p^b}(q)$ and $r \mid q-1$, 
by Lemma~\ref{lem:liord} we have
$l_i \mid p^{b}$ for any $i \in Z_{n,r}$.

For $j=0,1,\ldots,b-1$,  we have
\begin{equation}  \label{eq:nmodpj}
\begin{split}
    n=\frac{q^{p^{b}}-1}{q^{p^{b-1}}-1}&=q^{p^{b-1}(p-1)}+q^{p^{b-1}(p-2)}+\cdots+q^{p^{b-1}}+1\\
    &\equiv 1+1+\cdots+1+1 \equiv p \pmod{q^{p^j}-1}, 
\end{split}
\end{equation}
which implies 
$$
\gcd(n,q^{p^j}-1)=\gcd(p,q^{p^j}-1).
$$
If $p \mid q$, clearly we have $\gcd(p,q^{p^j}-1)=1$. 
In addition, if $p \nmid q$, suppose $\gcd(p,q^{p^j}-1)=p$, then we have $p \mid q-1$ 
(because $p \mid q^{p^j}-1$ and $p \mid q^{p-1}-1$), which contradicts with the assumption $p \nmid q-1$, 
and thus, we must have  $\gcd(p,q^{p^j}-1)=1$. 
Therefore,  we always have
\begin{equation}  \label{eq:gcdnqpjp}
  \gcd(n,q^{p^j}-1)=\gcd(p,q^{p^j}-1)=1 \quad \textrm{for} \quad j=0,1,\ldots,b-1, 
\end{equation}
which also implies (noticing $r \mid q-1$)
\begin{equation}  \label{eq:gcdrnqpj}
 \gcd\left(r, \frac{nr}{\gcd(nr,q^{p^j}-1)}\right)= \gcd\left(r, \frac{n}{\gcd(n,(q^{p^j}-1)/r)}\right)=1.
\end{equation}

  Thus, applying Lemma~\ref{lem:Nl2consta} with \eqref{eq:gcdnqpjp} and \eqref{eq:gcdrnqpj} and noticing the definition of the M\"{o}bius function,
  we obtain 
  \begin{align*}
    N_1&=\gcd(n,(q-1)/r) = 1,\\
    N_{p^{b}}&=\Big(\sum_{j=0}^{b} \mu(p^{b-j}) \gcd(n,(q^{p^j}-1)/r) \Big)/p^b  \\
     &=\frac{\gcd(n,(q^{p^{b}}-1)/r)-\gcd(n,(q^{p^{b-1}}-1)/r)}{p^{b}} 
=\frac{n- 1}{p^{b}}.
  \end{align*}
Moreover, noticing $N_1 + p^b N_{p^b}=n=|Z_{n,r}|$, we must have 
$
  \{l_i: \, i\in Z_{n,r}\} = \{1, p^b\}.
$
Then, the desired result follows directly from  Theorems~\ref{thm:NlNm} and \ref{thm:NlNm2}. 
\end{proof}

Now, we consider the case when $p \mid \frac{q-1}{r}$. 

\begin{theorem}\label{thm:pb2}
  Assume $p \mid (q-1)/r$ and $n=\Phi_{p^{b}}(q)$.
  Then,  an infinite family of $q$-ary $[n,n-\lceil p/2 \rceil-p^{b}\lceil \frac{n-p}{2p^{b}} \rceil,\ge \left\lceil \frac{q \lceil (n-p)/(2p^b) \rceil}{q-1} \right\rceil + \Lambda_{p^b}]$ constacyclic  codes is constructed, 
where $\Lambda_{p^b}$ has been defined in Theorem~\ref{thm:NlNm2}.  
  Moreover, if either $r \le 3$ or $2p^{b}(q-1) \ge qr$,
  an infinite family of $q$-ary $[n,n-\lfloor p/2 \rfloor -p^{b} \lfloor\frac{n-p}{2p^{b}} \rfloor,\ge \lfloor \frac{q(n-p)}{2p^{b}(q-1)} \rfloor]$ constacyclic codes is constructed.
\end{theorem}

\begin{proof}
  Combining \eqref{eq:nmodpj} with the assumption $p \mid (q-1)/r$, 
we have that for $j=0, 1, \ldots, b-1$, 
\begin{equation}  \label{eq:npjrp}
  \gcd(n,(q^{p^j}-1)/r)=\gcd(p,(q^{p^j}-1)/r)=p. 
\end{equation}

In addition, by \eqref{eq:nmodpj} and noticing $r \mid q-1$, we have 
$
\gcd(n,r) = \gcd(p,r). 
$
So, $\gcd(n,r)=1$ or $p$. 
If $\gcd(n,r)=1$, clearly we have $\gcd(n/p,r)=1$.
If $\gcd(n,r)=p$, then noticing $p\mid (q-1)/r$, we can write $q=ap^2+1$ for some positive integer $a$, and so, 
\begin{equation*} 
\begin{split}
    n &=q^{p^{b-1}(p-1)}+q^{p^{b-1}(p-2)}+\cdots+q^{p^{b-1}}+1\\
& = (ap^2+1)^{p^{b-1}(p-1)}+ \cdots + (ap^2+1)^{p^{b-1}} + 1 \\
&\equiv p \pmod{p^2}, 
\end{split}
\end{equation*}
which implies $p \nmid n/p$. 
Thus, when $\gcd(n,r)=p$, we have $p \nmid n/p$, and so $\gcd(n/p,r)=1$. 
Hence, we always have 
\begin{equation}   \label{eq:npr1}
\gcd(n/p,r)=1.
\end{equation} 
Therefore, combining \eqref{eq:npjrp} with \eqref{eq:npr1}, we get that 
for $j=0, 1, \ldots, b-1$, 
\begin{equation}  \label{eq:rnrqpj}
\begin{split}
  \gcd(r,\frac{nr}{\gcd(nr,q^{p^j}-1)})  = \gcd(r,\frac{n}{\gcd(n,(q^{p^j}-1)/r)}) 
 =  \gcd(r,n/p)=1. 
\end{split}
\end{equation}

Now, applying Lemma~\ref{lem:Nl2consta} with \eqref{eq:npjrp} and \eqref{eq:rnrqpj},
  we obtain 
  \begin{align*}
    N_1&=\gcd(n,(q-1)/r) = p,\\
    N_{p^{b}}&=\Big(\sum_{j=0}^{b} \mu(p^{b-j}) \gcd(n,(q^{p^j}-1)/r) \Big)/p^b  \\
     &=\frac{\gcd(n,(q^{p^{b}}-1)/r)-\gcd(n,(q^{p^{b-1}}-1)/r)}{p^{b}} 
=\frac{n- p}{p^{b}}.
  \end{align*}
Moreover, noticing $N_1 + p^b N_{p^b}=n=|Z_{n,r}|$, we must have 
$
  \{l_i: \, i\in Z_{n,r}\} = \{1, p^b\}.
$
Then, the desired result follows directly from Theorems~\ref{thm:NlNm} and \ref{thm:NlNm2}. 
\end{proof}

We remark that in Theorems~\ref{thm:pb1} and \ref{thm:pb2}, since $N_1 >0$, 
by Remark~\ref{rem:cyclic} we know that the constacyclic codes 
there are monomially equivalent to some cyclic codes over $\F_q$.

Finally, we consider the case when $p \mid q-1$ and $p \nmid (q-1)/r$.

\begin{theorem}\label{thm:pb3}
  Assume $p \mid q-1$, $p \nmid (q-1)/r$, and $n=\Phi_{p^{b}}(q)$. 
Assume further either $b \ge 2$  or $p$ is odd. 
  Then, 
  an infinite family of $q$-ary $[n,n-p-p^{b}\lceil \frac{n-p}{2p^{b}} \rceil,\ge \left\lceil \frac{q \lceil (n-p)/(2p^b) \rceil}{q-1} \right\rceil + \Lambda_{p^b}]$ constacyclic codes is constructed, 
where $\Lambda_{p^b}$ has been defined in Theorem~\ref{thm:NlNm2}.  
  Moreover, if either $r \le 3$ or $2p^{b-1}(q-1) \ge qr$,
  an infinite family of $q$-ary $[n,n-p^{b} \lfloor \frac{n-p}{2p^{b}} \rfloor,\ge \lfloor \frac{q(n-p)}{2p^{b}(q-1)} \rfloor]$ constacyclic codes is constructed.
\end{theorem}

\begin{proof}
First, since $p \mid q-1$ and $p \nmid (q-1)/r$, we must have $p\mid r$. 

Combining \eqref{eq:nmodpj} with the assumption $p \nmid (q-1)/r$, 
  we have 
\begin{equation}  \label{eq:nqrp1}
\gcd(n,(q-1)/r) = \gcd(p, (q-1)/r) = 1. 
\end{equation}
Similarly, combining \eqref{eq:nmodpj} with the assumption $p \mid q-1$,  for $j=1, \ldots, b-1$ we have 
\begin{equation}  \label{eq:nqpjrp}
\begin{split}
    \gcd(n,(q^{p^j}-1)/r) & = \gcd(p,(q^{p^j}-1)/r) \\
& = \gcd(p,\frac{q-1}{r}(q^{p^j-1}+ \cdots +q+1)) \\
& = \gcd(p,q^{p^j-1}+ \cdots +q+1) = p.
\end{split}
\end{equation}

In addition, since $p\mid q-1$, we write $q=ap+1$. 
Then,  
\begin{equation*} 
\begin{split}
    n &=q^{p^{b-1}(p-1)}+q^{p^{b-1}(p-2)}+\cdots+q^{p^{b-1}}+1\\
& = (ap+1)^{p^{b-1}(p-1)}+ \cdots + (ap+1)^{p^{b-1}} + 1 \\
    &\equiv (1+p^{b-1}(p-1)ap)+\cdots+(1+p^{b-1}ap)+1 \\
&\equiv p + ap^b \cdot \frac{p(p-1)}{2} \equiv p \pmod{p^2}, 
\end{split}
\end{equation*}
where the last congruence follows from the assumption that $b \ge 2$ or $p$ is odd. 
So, we have $p \nmid n/p$. 
In addition, using \eqref{eq:nmodpj} and noticing $p \mid r$, we get 
\begin{equation}  \label{eq:gcd-npr}
\gcd(n,r)=\gcd(p,r)=p.
\end{equation}
Hence, we must have 
\begin{equation}  \label{eq:npr=1}
\gcd(n/p,r)=1.
\end{equation}

Then, using \eqref{eq:nqrp1} and \eqref{eq:gcd-npr}, we obtain 
\begin{equation}  \label{eq:rnrq=p}
\begin{split}
    \gcd\left(r,\frac{nr}{\gcd(nr,q-1)}\right)&=\gcd\left(r,\frac{n}{\gcd(n,(q-1)/r)}\right)\\
    &=\gcd(r,n) =p.
\end{split}
\end{equation}
For $j=1, \ldots, b-1$, using \eqref{eq:nqpjrp} and  \eqref{eq:npr=1}, we have 
\begin{equation}   \label{eq:rnrqpj=1}
\begin{split}
    \gcd\left(r,\frac{nr}{\gcd(nr,q^{p^j}-1)}\right)&=\gcd\left(r,\frac{n}{\gcd(n,(q^{p^j}-1)/r)}\right)\\
    &=\gcd(r,n/p) = 1.
\end{split}
\end{equation}

Now, applying Lemma~\ref{lem:Nl2consta} with \eqref{eq:nqpjrp}, \eqref{eq:rnrq=p} and 
\eqref{eq:rnrqpj=1},
  we obtain 
  \begin{align*}
N_1 & = 0, \\
    N_p&=\frac{\gcd(n,(q^p-1)/r)-0}{p}= \frac{p-0}{p}=1,\\
    N_{p^{b}}&=\Big(\sum_{j=0}^{b} \mu(p^{b-j}) \gcd(n,(q^{p^j}-1)/r) \Big)/p^b  \\
     &=\frac{\gcd(n,(q^{p^{b}}-1)/r)-\gcd(n,(q^{p^{b-1}}-1)/r)}{p^{b}} 
=\frac{n- p}{p^{b}}.
  \end{align*}
Moreover, noticing $pN_p + p^b N_{p^b}=n=|Z_{n,r}|$, we must have 
$
  \{l_i: \, i\in Z_{n,r}\} = \{p, p^b\}.
$
Then, the desired result follows directly from  Theorems~\ref{thm:NlNm} and \ref{thm:NlNm2}.
\end{proof}

We remark that in Theorem~\ref{thm:pb3}, from \eqref{eq:nqrp1} and \eqref{eq:gcd-npr} 
we have $\gcd(n, q-1)=p$, which does not divide $(q-1)/r$ by assumption; 
and so, the sufficient condition in Lemma~\ref{lem:equivalence} 
about equivalence between constacyclic codes and cyclic codes fails in this case. 

We notice that some codes in this section are better than BCH constacyclic codes; see
Example~\ref{ex:pb1}.

\begin{example}  \label{ex:pb1}
Let $q=4$, $p=2$, $b=2$, $r=3$, and $n=17$. 
  Then, the $4$-cyclotomic cosets in $Z_{17, 3}$ are: 
  $C_{1}=\{1, 4, 13, 16\}$, 
  $C_{7}=\{7, 10, 28, 40\}$, 
  $C_{19}=\{19, 25, 43, 49\}$, 
  $C_{22}=\{22, 31, 37, 46\}$, 
  $C_{34}=\{34\}$.
So, $\{l_i,i\in Z_{17,3}\}=\{1,4\}$, $N_1=1$ and $N_4 = 4$. 
From the first family of constacyclic codes in Theorem~\ref{thm:pb1}, we take the defining set as the set  $C_1 \cup C_7 \cup C_{34}$ and then obtain a $4$-ary constacyclic $[17, 8, 8]$ code, which is not a BCH code and has the best-known parameters \cite{Grassl}. 
In this case, there is only one BCH constacyclic code of length 17 and dimension 8
(its defining set is $C_7 \cup C_{22} \cup C_{34}$ and minimum distance is 6).  
\end{example}

In Table~\ref{tab:qpb}, we present some codes constructed in Section~\ref{sec:pbq} having good parameters.
In the table, 
$d_{\text{best}}$ stands for the maximal minimum distance of all known $q$-ary linear codes with relevant length and dimension, and its value is from \cite{Grassl}.  

\begin{longtable}{ccccccccc}    
  \caption{Codes in Section~\ref{sec:pbq}}
  \label{tab:qpb}\\
    \hline
    $q$&$p$&$b$&$r$ & Lower bound & Bose distance &Parameter&Optimality & References\\
    \hline
    \endfirsthead
    \hline
    $q$&$p$&$b$&$r$ & Lower bound & Bose distance &Parameter&Optimality & References\\
    \hline
    \endhead
    \hline
  \endfoot
    $2$&$2$&$3$&$1$&$2$&$3$&$[17,8,6]^*$&optimal&Theorem~\ref{thm:pb1}\\
    $2$&$2$&$3$&$1$&$2$&$3$&$[17,9,5]$&optimal&Theorem~\ref{thm:pb1}\\
    $2$&$3$&$2$&$1$&$8$&$11$&$[73,37,13]$&$d_\text{best}=14$&Theorem~\ref{thm:pb1}\\
    $4$&$2$&$2$&$3$&$3$&$7$&$[17,8,8]^*$&best-known&Theorem~\ref{thm:pb1}\\
    $4$&$2$&$2$&$3$&$2$&$7$&$[17,9,7]$&best-known&Theorem~\ref{thm:pb1}\\
    $5$&$2$&$2$&$2$&$4$&$5$&$[26,13,8]^*$&$d_\text{best}=10$&Theorem~\ref{thm:pb2}\\
    $5$&$2$&$2$&$4$&$4$&$9$&$[26,12,10]^*$&$d_\text{best}=11$&Theorem~\ref{thm:pb3}\\
    $5$&$2$&$2$&$4$&$3$&$4$&$[26,14,8]$&$d_\text{best}=9$&Theorem~\ref{thm:pb3}\\
    $7$&$2$&$2$&$2$&$7$&$9$&$[50,26,14]$&$d_\text{best}=15$&Theorem~\ref{thm:pb3}\\
    $8$&$2$&$2$&$1$&$10$&$22$&$[65,32,22]^*$&best-known&Theorem~\ref{thm:pb1}\\
     \end{longtable}

\section{Constacyclic codes with $n=\Phi_{p_1 p_2}(q)$}   \label{sec:p1p2}

Recall that $q$ is a prime power, and $r\mid q-1$.
We assume that $p_1,p_2$ are two distinct primes with $p_1 < p_2$.
In the following, we construct constacyclic codes with length
$$
n=\Phi_{p_1 p_2}(q)= \frac{(q^{p_1 p_2} -1)(q-1)}{(q^{p_1}-1)(q^{p_2}-1)} = \frac{q^{p_1 p_2} -1}{\lcm(q^{p_1}-1,q^{p_2}-1)}.
$$

Moreover, we only consider the case when the set  $\{l_i: \, i\in Z_{n,r}\}$, 
where $l_i = |C_i|$, has exactly two elements, and then 
 we use Theorem~\ref{thm:NlNm} to construct constacyclic codes having good lower bounds for their minimum distances.

\begin{theorem}\label{thm:p1p2}
  Let $p_1,p_2$ be two distinct primes with $p_1<p_2$, and 
  $n=\Phi_{p_1 p_2}(q)$. Assume 
  $p_2 \nmid (q^{p_1}-1)/r$.  
  Then, an infinite family of $q$-ary $[n,n-1-p_1 p_2 \lceil(n-1)/(2p_1 p_2) \rceil,\ge \left\lceil \frac{q \lceil (n-1)/(2p_1p_2) \rceil}{q-1} \right\rceil + \Lambda_{p_1p_2}]$ constacyclic codes is constructed, 
where $\Lambda_{p_1p_2}$ has been defined in Theorem~\ref{thm:NlNm2}. 
  Moreover, if either $r \le 3$ or $2 p_1 p_2 (q-1) \ge qr$,
  an infinite family of $q$-ary $[n,n- p_1 p_2 \lfloor(n-1)/(2p_1 p_2) \rfloor,\ge\lfloor \frac{q(n-1)}{2p_1 p_2(q-1)} \rfloor]$ constacyclic codes is constructed. 
\end{theorem}

\begin{proof}
  First, we want to determine the set $\{l_i: \, i\in Z_{n,r}\}$, where $l_i = |C_i|$. 
Since $n=\Phi_{p_1 p_2}(q)$ and $r \mid q-1$,  
by Lemma~\ref{lem:liord} 
we have
  $l_i \mid p_1p_2$ for any $i \in Z_{n,r}$.

  By Lemma~\ref{lem:gcdPhiq-1},
  we have 
\begin{equation}  \label{eq:gcdnq1}
\gcd(n,q-1)= \gcd(\Phi_{p_1 p_2}(q),q-1) = 1.
\end{equation}
  This, together with $r \mid q-1$, implies that for any integer $j \ge 1$,
  we have 
\begin{equation}  \label{eq:rnrqj=1}
  \gcd\left(r, \frac{nr}{\gcd(nr,q^j-1)}\right)= \gcd\left(r, \frac{n}{\gcd(n,(q^j-1)/r)}\right)=1.
\end{equation}

Notice that
  \begin{align*}
    \frac{q^{p_1 p_2}-1}{q^{p_1}-1}
    &=q^{p_1(p_2-1)}+q^{p_1(p_2-2)}+\cdots+q^{p_1}+1\\
    &\equiv 1+1+\cdots+1+1 \equiv p_2 \pmod{q^{p_1}-1}. 
  \end{align*}
  Then, we get 
  $$
  \gcd\left(\frac{q^{p_1 p_2}-1}{q^{p_1}-1},\frac{q^{p_1}-1}{r}\right)=\gcd\left(p_2,\frac{q^{p_1}-1}{r}\right)=1, 
  $$
where the last equality follows from the assumption $p_2 \nmid (q^{p_1}-1)/r$. 
So, noticing  $\frac{q^{p_1 p_2}-1}{q^{p_1}-1} = n \cdot \frac{q^{p_2}-1}{q-1}$, we obtain 
\begin{equation}  \label{eq:nqp1r}
\gcd(n,(q^{p_1}-1)/r) = 1.
\end{equation}

Now,  applying Lemma~\ref{lem:Nl2consta} with \eqref{eq:gcdnq1}, \eqref{eq:rnrqj=1} and \eqref{eq:nqp1r},
  we obtain
  \begin{align*}
    N_1&=\gcd(n,(q-1)/r) = 1,\\
      N_{p_1}&=\left(\gcd\left(n,\frac{q^{p_1}-1}{r}\right)-N_1\right) / p_1=(1-1)/p_1=0.
  \end{align*}
  
  Next, we want to compute $N_{p_2}$. 
Similarly as deducing \eqref{eq:nqp1r}, we have 
\begin{equation*}
\gcd\left(n,\frac{q^{p_2}-1}{r}\right) 
\le \gcd\left(\frac{q^{p_1 p_2}-1}{q^{p_1}-1},\frac{q^{p_2}-1}{r}\right)
=  \gcd\left(p_1,\frac{q^{p_2}-1}{r}\right) \le p_1.
\end{equation*}
  Combining this with $N_1=1$ and the assumption $p_1 < p_2$,
  we get 
  $$
   N_{p_2}=\left(\gcd\left(n,\frac{q^{p_2}-1}{r}\right)-N_1 \right)/p_2\le (p_1-1)/p_2 < 1.
  $$
  Thus, we have $N_{p_2}=0$.

  In addition,   we directly have
  $$
  N_{p_1 p_2}=\frac{n-N_1-p_1 N_{p_1}-p_2 N_{p_2}}{p_1 p_2}=\frac{n-1}{p_1 p_2}.
  $$
Hence, we have 
$
  \{l_i: \, i\in Z_{n,r}\} = \{1, p_1p_2\}.
$
Then, the desired results follow directly from Theorems~\ref{thm:NlNm} and \ref{thm:NlNm2}.
\end{proof}

We remark that in Theorem~\ref{thm:p1p2}, since $N_1 >0$, 
by Remark~\ref{rem:cyclic} we know that the constacyclic codes 
there are monomially equivalent to some cyclic codes over $\F_q$.

We also remark that the condition $p_2 \nmid (q^{p_1}-1)/r$ in Theorem~\ref{thm:p1p2} holds when $p_1 \nmid p_2-1$. 
Indeed, if $p_1 \nmid p_2-1$, then $q^{p_1} \not\equiv 1 \pmod{p_2}$, and so $p_2 \nmid q^{p_1}-1$.  

We notice that some codes in Theorem~\ref{thm:p1p2} are better than BCH constacyclic codes; see
Example~\ref{ex:p1p2}. 

\begin{example}  \label{ex:p1p2}
Let $q=7$, $p_1=2$, $p_2=3$, $r=2$, and $n=43$. 
  Then, the $7$-cyclotomic cosets in $Z_{43, 2}$ are: 
$C_{1}=\{1, 7, 37,49, 79, 85\}$, 
  $C_{3}=\{3, 21, 25, 61, 65, 83\}$, 
  $C_{5}=\{5, 13, 35, 51, 73, 81\}$, 
  $C_{9}=\{9, 11, 23, 63, 75, 77\}$, 
  $C_{15}=\{15, 19, 39, 47, 67, 71\}$,
  $C_{17}=\{17, 27, 33, 53, 59, 69\}$,
  $C_{29}=\{29, 31, 41, 45, 55, 57\}$,
  $C_{43}=\{43\}$.
So, $\{l_i,i\in Z_{43,2}\}=\{1,6\}$, $N_1=1$ and $N_6 = 7$. 
From the first family of constacyclic codes in Theorem~\ref{thm:p1p2}, we take the defining set as the set $C_1 \cup C_3 \cup C_5 \cup C_9 \cup C_{43}$ and then obtain a $7$-ary constacyclic $[43, 18, 17]$ code, which is not a BCH constacyclic code ($d_{\text{best}}=18$ according to \cite{Grassl}). 
In this case, there is only one BCH constacyclic code of length 43 and dimension 18
(its defining set is $C_1 \cup C_5 \cup C_{15} \cup C_{29} \cup C_{43}$ and minimum distance is 16).  
\end{example}

In Table~\ref{tab:p1p2}, we present some codes constructed in Theorem~\ref{thm:p1p2} having good parameters.

\begin{longtable}{cccccccc}    
  \caption{Codes in Theorem \ref{thm:p1p2}}
  \label{tab:p1p2}\\
    \hline
    $q$&$p_1$&$p_2$&$r$ & Lower bound & Bose distance &Parameter&Optimality\\
    \hline
    \endfirsthead
    \hline
    $q$&$p_1$&$p_2$&$r$ & Lower bound & Bose distance &Parameter&Optimality\\
    \hline
    \endhead
    \hline
  \endfoot
    $2$&$2$&$7$&$1$&$4$&$7$&$[43,14,14]^*$&optimal\\
    $2$&$2$&$7$&$1$&$3$&$3$&$[43,29,6]$&optimal\\
    $3$&$2$&$5$&$1$&$4$&$5$&$[61,31,14]$&$d_\text{best}=16$\\
    $4$&$2$&$3$&$3$&$2$&$3$&$[13,6,6]^*$&optimal\\
    $4$&$2$&$3$&$3$&$1$&$3$&$[13,7,5]$&optimal\\
   $7$&$2$&$3$&$2$&$5$&$8$&$[43,18,17]^*$&$d_\text{best}=18$\\
    $7$&$2$&$3$&$3$&$4$&$5$&$[43,25,11]$&$d_\text{best}=12$\\
    $7$&$2$&$3$&$6$&$5$&$7$&$[43,18,15]^*$&$d_\text{best}=18$\\
    \end{longtable}

\section{Concluding remarks}
We constructed several infinite families of 
constacyclic codes over $\F_q$ in this paper. 
They contain many constacyclic codes with optimal, or almost-optimal, or best-known parameters. 
We presented some of them in Tables~\ref{tab:qp}, \ref{tab:qpb} and \ref{tab:p1p2}, 
and the codes in these tables are new compared to the constacyclic codes in the references below.
Especially, we considered two forms of the code length: $n=\Phi_{p^b}(q)$ or $\Phi_{p_1p_2}(q)$, 
which are new up to our knowledge. 
We also presented several codes which are better than BCH constacyclic codes (see Examples~\ref{ex:qp}, \ref{ex:pb1} and \ref{ex:p1p2}).

For the construction, we first classified the cyclotomic cosets according to their sizes, 
and then we took the defining set as the union of the first half of cyclotomic cosets in each class 
(sorted in ascending order by their coset leaders). 
In this paper, we focused on 
 the case when there are exactly two kinds of cyclotomic cosets.  
According to our Magma experimental data, when there are more than two kinds of cyclotomic cosets, 
there are still some good codes under the above construction. 

We emphasize that the codes we constructed are not always BCH constacyclic codes.  
The other fact we want to emphasize is that if 
 $N$ is the number of cyclotomic cosets with maximum size, then
in the above construction, the first $\lfloor N/2 \rfloor$ consecutive elements
may be not always in the defining set (see Example~\ref{ex:deltaN1Nm}). 
This might deserve further investigation. 

For the specific codes presented in Sections~\ref{sec:qp}, 
\ref{sec:pbq} and \ref{sec:p1p2}, except the codes in Theorem~\ref{thm:pb3}, 
the others have been proved to be monomially equivalent to some cyclic codes over $\F_q$. 
By looking at the defining sets of these cyclic codes, we haven't found any simple rule which can give an infinite family of cyclic codes. 

We also remark that for the second family of constacyclic codes in Theorems~\ref{thm:qp}, \ref{thm:pb1}, \ref{thm:pb2} and \ref{thm:p1p2}, if furthermore $r \le 2$, then 
by Theorem~\ref{thm:NlNm2}, one can get a better lower bound on their minimum distances.

\end{document}